\documentclass[12pt]{article}
\usepackage{amsmath}
\usepackage{graphicx}
\usepackage{enumerate}
\usepackage{natbib}
\usepackage{url} 
\usepackage{lineno}
\usepackage{amssymb}
\usepackage{amsthm}
\usepackage{authblk}
\usepackage{epsfig}
\usepackage{graphics}
\usepackage{float}
\usepackage{subfigure}
\usepackage{multirow}
\usepackage{xcolor}
\usepackage{makeidx}
\usepackage{xspace}
\usepackage{wrapfig}
\usepackage{blindtext}
\usepackage{setspace}
\usepackage{algorithm}
\usepackage{algorithmic}
\usepackage[colorlinks=true,citecolor=blue,linkcolor=blue,urlcolor=blue]{hyperref}
\usepackage{booktabs}
\usepackage{bm}
\makeindex

\newtheorem{Theorem}{Theorem}

\definecolor{darkred}{rgb}{1, 0.1, 0.3}
\definecolor{darkblue}{rgb}{0.1, 0.1, 1}
\definecolor{darkgreen}{rgb}{0,0.6,0.5}

\newcommand{\blind}{1}

\begin{document}

\def\spacingset#1{\renewcommand{\baselinestretch}%
{#1}\small\normalsize} \spacingset{1}


\if1\blind
{
  \title{\bf Large-Scale Low-Rank Gaussian Process Prediction with Support Points}
  \author{Yan Song$^{1,2}$, Wenlin Dai$^1$\thanks{
    Corresponding author; email: wenlin.dai@ruc.edu.cn \textit{}}\hspace{.2cm}\\
    \vspace{-12pt}
    $^1$Institute of Statistics and Big Data, Renmin University of China, China\\
    and \\
    Marc G. Genton$^2$ \\
    $^2$Statistics Program, King Abdullah University of Science and Technology, Saudi Arabia}
  \maketitle
} \fi

\if0\blind
{
  \title{\bf Large-Scale Low-Rank Gaussian Process Prediction with Support Points}
  \author{ \\
     \vspace{30 mm}
     }
  \maketitle
} \fi


\begin{abstract}
Low-rank approximation is a popular strategy to tackle the ``big $n$ problem" associated with large-scale Gaussian process regressions. Basis functions for developing low-rank structures are crucial and should be carefully specified. Predictive processes simplify the problem by inducing basis functions with a covariance function and a set of knots. The existing literature suggests certain practical implementations of knot selection and covariance estimation; however, theoretical foundations explaining the influence of these two factors on predictive processes are lacking. In this paper, the asymptotic prediction performance of the predictive process and Gaussian process predictions are derived and the impacts of the selected knots and estimated covariance are studied. The use of support points as knots, which best represent data locations, is advocated. Extensive simulation studies demonstrate the superiority of support points and verify our theoretical results. Real data of precipitation and ozone are used as examples, and the efficiency of our method over other widely used low-rank approximation methods is verified.
\end{abstract}

\noindent%


\spacingset{1.9} 
\section{Introduction}
\label{sec:Introduction}
Gaussian processes \citep[GPs;][]{GP} are extensively used for solving regression problems in many fields such as spatial statistics \citep{Stein1999,HBS,Cressie2015}, computer experiments \citep{Sacks1989,Santner2003}, and machine learning \citep{GP,Liu2020}. The underlying regression function is assumed to be a realization of a GP, whose mean trend and covariance function must be evaluated. Using the conditional multivariate normal distribution, a GP provides each location with a closed-form prediction as well as a predictive interval. 

Assume that a dataset $D=\{(\mathbf{x}_i,y_i)\}_{i=1}^n$ satisfies the GP model, i.e.,
\vspace{-20pt}
\begin{equation}
    y_i=f(\mathbf{x}_i)+\epsilon_i,\quad i=1,\ldots,n,
    \label{eq:model}
    \vspace{-20pt}
\end{equation}
where $\mathbf{x}_i\in\mathcal{X}\subset\mathbb{R}^d$ and follows a distribution $F_{\mathbf{x}}$. The $f(\mathbf{x}_i)$'s are the realization of a Gaussian process $Z(\cdot)$ on locations $\{\mathbf{x}_i\}_{i=1}^n$, and the nugget effects $\epsilon_i$'s are independent and follow a $\mathcal{N}(0,\tau^2)$ distribution. Under the GP regression, $f$ is considered as a random function and is not distinguished from $Z$. Assuming $Z(\cdot)\sim\mathrm{GP}(0,c)$, where $c(\cdot,\cdot)$ is a positive definite covariance function. The zero mean assumption simplifies the presentation, as estimating and subtracting the mean is typically computationally feasible, and theoretical results for a general case can be obtained by detrending the process \citep{burt2019rates,Vecchia_Katzfuss2020}. Then, the response vector $\mathbf{y}=(y_1,\ldots,y_n)^{\top}$ follows a multivariate normal distribution $\mathcal{N}_n(\mathbf{0},\mathbf{C}+\tau^2\mathbf{I}_n)$, where $\mathbf{C}=\{c(\mathbf{x}_i,\mathbf{x}_j)\}_{i,j=1}^n$. The covariance function $c$ is usually unknown in practice and is represented by a specific parametric form $c(\cdot,\cdot;\boldsymbol{\theta})$, where the parameter vector $\boldsymbol{\theta}$ can be obtained by maximizing the likelihood function with respect to $\boldsymbol{\theta}$: $(2\pi)^{-n/2}|\mathbf{C}(\boldsymbol{\theta})+\tau^2\mathbf{I}_n|^{-1/2}\exp\left [-\frac{1}{2}\mathbf{y}^{\top}\{\mathbf{C}(\boldsymbol{\theta})+\tau^2\mathbf{I}_n\}^{-1}\mathbf{y}\right]$. We avoid theoretical discussions on covariance estimation due to its complexity and methodological dependency, opting instead to demonstrate how an estimated covariance impacts predictions.
With the conditional probability density of the multivariate normal distribution, the prediction based on $c$ at any $\mathbf{x}\in\mathcal{X}$ is
\vspace{-20pt}
\begin{equation}
    \hat f_{c}(\mathbf{x})=\mathbf{c}(\mathbf{x})^{\top} (\mathbf{C}+\tau^2 \mathbf{I}_n)^{-1}\mathbf{y},
    \label{eq:fullpred}
    \vspace{-10pt}
\end{equation}
where $\mathbf{c}(\mathbf{x})=\{c(\mathbf{x},\mathbf{x}_1),\ldots,c(\mathbf{x},\mathbf{x}_n)\}^\top$. Fitting data of size $n$ with a GP requires $O(n^3)$ operations and $O(n^2)$ memory, and these requirements may easily exhaust computational resources even with a moderately large $n$. 

Many methods are available to tackle the challenge of handling massive datasets with GPs. We refer the readers to \citet{heaton2017methods} and \citet{Liu2020} for their comprehensive reviews. The methods of likelihood approximation include the Vecchia approximation \citep{Vecchia1988,Stein2004,Vecchia_Katzfuss2020,Vecchiageneral}, composite likelihood \citep{composite,blockcomposite}, and hierarchical low-rank approximation \citep{Huanglikeli}. Covariance tapering multiplies the covariance function with a compactly supported one such that the tapered covariance matrix is sparse and can be solved easily \citep{taper2006,taper2008,taper2013}. \citet{Rue2002}, \citet{Rue2005}, and \citet{Markov_Marc2015} discussed the methods of Markov random field approximation, which are primarily used for lattice data. Methods based on distributed computing \citep{parallel2013,distribitedGP,Katzfuss2017,ExaGeoStat1,ExaGeoStat2,ExaGeoStat3}, resampling \citep{liang2013,barbian2017spatial}, and exact approaches using the \textit{ExaGeoStat} software \citep{ExaGeoStat1} were recently proposed.

The low-rank approximation is another important approach for GPs with massive data. It approximates the original GP of mean zero with a low-rank one, which is represented by a linear combination of specified basis functions with random weights, e.g., the fixed-rank kriging \citep[FRK,][]{FRK2008,FRK_intro}, LatticeKrig \citep{LK2015}, and predictive processes \citep{PP,PPimprove}. Combining low-rank approximations with the concept of covariance tapering can help develop multiple efficient methods \citep{sang2013,MRA2017}. In this context, the choice of the basis functions is important to the performance of these methods. FRK and LatticeKrig suggest multiresolution bases to capture information at different scales. The construction of basis functions involves the delicate selections of multiple tuning parameters, such as the types of basis functions, levels of resolution, and information in the coarsest resolution \citep{FRK_intro,LKPackage}. \citet{autoFRK} proposed adaptive basis functions to avoid the manual allocation of scales. 

Predictive processes simplify the aforementioned problem by developing basis functions with a covariance function and a set of specified knots. In this context, the covariance estimation and knot configuration are of particular interest. Specifically, \citet{PP} selected a set of knots $D_k^*=\{\mathbf{x}_i^*\}_{i=1}^k$ well dispersed over $\mathcal{X}$, which may or may not belong to $\{\mathbf{x}_i\}_{i=1}^n$. Then, they projected the original process $Z$ onto a linear space spanned by $\mathbf{z}^*=\{Z(\mathbf{x}_1^*),\ldots,Z(\mathbf{x}_k^*)\}^\top$ and induced the predictive process $\tilde{Z}(\mathbf{x})=\mathbb{E}\{Z(\mathbf{x})|\mathbf{z}^*\}=\mathbf{c}^*(\mathbf{x})^\top \mathbf{C}^{*-1}\mathbf{z}^*$, where $\mathbf{c}^*(\mathbf{x})=\{c(\mathbf{x},\mathbf{x}_1^*),\ldots,c(\mathbf{x},\mathbf{x}_k^*)\}^\top$ and $\mathbf{C}^*=\{\mathbf{c}^*(\mathbf{x}_1^*),\ldots,\mathbf{c}^*(\mathbf{x}_k^*)\}$. Correspondingly, the process $\tilde{Z}\sim\mathrm{GP}(0,\tilde{c})$, where  $\tilde{c}(\mathbf{x}_i,\mathbf{x}_j)=\mathbf{c}^*(\mathbf{x}_i)^\top \mathbf{C}^{*-1}\mathbf{c}^*(\mathbf{x}_j)$. Replacing $c$ in \eqref{eq:fullpred} with $\tilde{c}$ yields the low-rank prediction 
\vspace{-20pt}
\begin{equation}
    \tilde{f}_c(\mathbf{x})=\mathbf{c}^*(\mathbf{x})^\top \mathbf{C}^{*-1} \mathbf{C}_{nk}^{*\top}(\mathbf{C}_{nk}^* \mathbf{C}^{*-1}\mathbf{C}_{nk}^{*\top}+\tau^2 \mathbf{I}_n)^{-1}\mathbf{y},
    \label{eq:pp}
    \vspace{-20pt}
\end{equation}
where $\mathbf{C}_{nk}^*=\{\mathbf{c}^*(\mathbf{x}_1),\ldots,\mathbf{c}^*(\mathbf{x}_n)\}^\top$. Furthermore, using the Sherman--Morrison--Woodbury formula \citep{smw}, we can rewrite \eqref{eq:pp} as follows: \\
\vspace{-20pt}
\begin{equation*}
    \tilde{f}_c(\mathbf{x})=\mathbf{c}^*(\mathbf{x})^\top \mathbf{C}^{*-1} \mathbf{C}_{nk}^{*\top}\{\tau^{-2}\mathbf{I}_n-\tau^{-2}\mathbf{C}_{nk}^{*}(\tau^2\mathbf{C}^*+\mathbf{C}_{nk}^{*}\mathbf{C}_{nk}^{*\top})^{-1}\mathbf{C}_{nk}^{*\top}\}\mathbf{y},
    \vspace{-10pt}
\end{equation*}
which includes the inversion of only a $k\times k$ matrix $\tau^2\mathbf{C}^*+\mathbf{C}_{nk}^{*}\mathbf{C}_{nk}^{*\top}$; hence, this step reduces the computational time to $O(nk^2)$. Finally, \citet{PP} implemented the entire procedure of parameter estimation and prediction by adopting the Bayesian approach. \citet{PPimprove} proposed a modified predictive process to debias the overestimation of $\tau^2$. Moreover, inspired by strategies in spatial designs, they developed an algorithm to select those knots that induce a predictive process mimicking the parent process better. Predictive processes are also termed as sparse Gaussian processes in the machine learning literature, where knots $D_k^*$ are termed as inducing points \citep{titsias2009variational,hensman2013gaussian,burt2019rates,burt2020convergence}.  

The idea of low-rank approximations is not limited to the GP model and has been applied to more problems, whereas the basis selection is still being discussed. The Nystr\"{o}m approximation shares some commonality with the low-rank approximation. It approximates a positive definite matrix $\mathbf{A}\in\mathbb{R}^{n\times n}$ by a low-rank one  $\tilde{\mathbf{A}}=\mathbf{AS}(\mathbf{S}^{\top}\mathbf{A}\mathbf{S})^{-}\mathbf{S}^{\top}\mathbf{A}$, where the ``sketching matrix" $\mathbf{S}\in\mathbb{R}^{n\times m}$, with $m<n$, sketches the desirable information from $\mathbf{A}$ and $\mathbf{K}^{-}$ is the generalized inverse of a matrix $\mathbf{K}$ \citep{Nystrom1,Nystrom2,Nystrom3}. This strategy has been used in the GP \citep{efficient2013} and the kernel ridge regression (KRR) models \citep{Rudi2015,Fastkernel}. The matrix $\mathbf{A}$ is prespecified and $\mathbf{S}$ determines the performance of the approximation.

The KRR is closely linked to GPs. In particular, $\hat f_c$ is a solution of a KRR obtained by considering $f$ in \eqref{eq:model} as a deterministic function:
\vspace{-20pt}
\begin{equation}
    \hat f_c=\arg\min_{\eta\in\mathcal{H}_c}\frac{1}{n}\sum_{i=1}^n\{y_i-\eta(\mathbf{x}_i)\}^2+\frac{\tau^2}{2n}J_c(\eta),
\label{eq:krr}
\vspace{-15pt}
\end{equation}
where $\mathcal{H}_c=\{\eta=\sum_{r=1}^{\infty}\alpha_r c(\cdot,\mathbf{x}_r): \alpha_r\in\mathbb{R}, \mathbf{x}_r\in\mathcal{X} \text{ such that }  \|\eta\|_{\mathcal{H}_c}^2=\sum_{r,l=1}^{\infty}\alpha_r\alpha_l c(\mathbf{x}_r,\mathbf{x}_l)<\infty\}$ is the reproducing kernel Hilbert space induced by $c$ with norm $J_c(\eta)=\|\eta\|_{\mathcal{H}_c}^2$. Wahba's representer theorem \citep{wahba1990spline} ensures that $\hat f_c$ lies in a finite subspace of $\mathcal{H}_c$ spanned by $\{c(\cdot,\mathbf{x}_i)\}_{i=1}^n$, i.e., $\mathcal{H}_{n,c}$. The low-rank prediction $\tilde{f}_c$ is the solution of \eqref{eq:krr} over $\mathcal{H}_{k,c}^*=\mathrm{span}\{c(\cdot,\mathbf{x}_i^*)\}_{i=1}^k$, induced by knots $D_k^*$. Using a prespecified kernel function $c$, the performance of the low-rank prediction $\tilde{f}_c$ depends on the choice of knots. 

There is a considerable body of literature discussing the selection of knots. \citet{PP} and \citet{PPimprove} drew inspiration from spatial designs, such as space-filling, to select knots. \citet{titsias2009variational} proposed a variational formulation that simultaneously estimates the covariance and selects knots by minimizing the Kullback-Leibler (KL) divergence between the variational distribution and the exact posterior distribution of the latent function values. However, the performance of the predictive process with these knots was only numerically tested without theoretical justification. The influence of knots on the asymptotic properties of $\tilde{f}_{\hat c}$ should be investigated. \citet{Rudi2015} and \citet{Fastkernel} chose knots that improve the prediction accuracy of $\tilde{f}_{c}$. They derived the optimal probabilities that minimize the learning bounds, which are based on a concept called ridge leverage scores (RLS). However, they used a random sampling strategy to obtain knots, which brings randomness and larger uncertainty. \citet{burt2019rates} and \citet{burt2020convergence} investigated how $k$ needs to grow with $n$ so that the KL divergence between the approximate model and the exact posterior can be arbitrarily small. However, they only focused on knots of size $k$ sampled from a fixed-size determinantal point process ($k$-DPP) and using RLS without considering the influence of a mis-specified covariance.

In this study, we investigate the asymptotic performance of the predictive process and the original process upon linking GP with KRR. Given the true covariance function $c$, we demonstrate that the predictive process may achieve the same convergence rate as the original GP. The convergence rate is determined by a parameter, $\gamma$, which is proposed as a measure of smoothness for the process and for the covariance function. Consequently, we derive a knot-selection criterion, i.e., support points generated by the data locations, which are space-filling, and the best approximation of the locations with respect to the energy distance. Moreover, we demonstrate how the estimated covariance function, $\hat c$, affects the convergence rates of the two processes through the smoothness parameter $\hat\gamma$ for $\hat c$. As a byproduct, the adjusted order of $\tau^2$ is obtained to improve the convergence rates of $\tilde{f}_{\hat c}$ and $\hat f_{\hat c}$. Extensive numerical studies are conducted to demonstrate the superior performance of the support points and validate our theoretical results. We claim that the convergence rate of $\tilde{f}_{\hat c}$ can be identical to that of $\hat f_{\hat c}$ if $D_k^*$ is close enough to $\{\mathbf{x}_i\}_{i=1}^n$. Thus, the efficiency of the low-rank approximation is validated.


The remainder of this paper is organized as follows. In Section~\ref{sec:method}, we provide theoretical results and describe the influence of knots and estimated covariance functions on the prediction performance of a predictive process. In Section~\ref{sec:simulation}, we present results of our numerical studies to demonstrate our approach and confirm the theoretical results. In Section~\ref{sec:realdata}, we describe the application of various low-rank approximations to two real data examples and compare their results. In Section~\ref{sec:discussion}, we present the conclusions and discuss future work. The proofs of the theoretical results are provided in the Supplementary Materials.

\section{Convergence Rates of Low-Rank Approximations with Support Points}
\label{sec:method}
\subsection{Known covariance function}
\label{sec:reppoints}
Let $D_k^*$ be a set of representative points of $\{\mathbf{x}_i\}_{i=1}^n$ used to develop a solution space for searching the low-rank approximation, hereinafter referred to as {\em rep-points}. With certain regularity conditions and a known covariance function $c$, the preferred configuration and size of $D_k^*$ for obtaining $\tilde{f}_c$ with the asymptotically optimal prediction performance are provided in this section.

The selection of $D_k^*$ has been extensively discussed. \citet{PP} tested the performances of three strategies, including the use of grids, lattice plus close pairs, and lattice plus infill design, for both parameter estimation and prediction. The latter two options were inspired from the concepts of spatial design \citep{diggle2006bayesian}. Compared to the configuration, the size of $D_k^*$ has greater influence on the parameter estimation. \citet{PPimprove} were motivated by spatial design concepts such as the space-filling design \citep{Nychka1998} and other optimal designs for various targets \citep{zhu2005spatial,xia2006}. They developed an algorithm for sequentially selecting $D_k^*$ such that the predictive process can approximate the parent process better. Given a kernel (covariance) function, \citet{Rudi2015} and \citet{Fastkernel} proposed randomly selecting subsamples from the full data as per the RLS. Due to the inherent randomness, their $D_k^*$ did not exhibit a space-filling property or a preference for certain points. In general, rep-points that are well dispersed over the region \citep{heaton2017methods} are preferable. Theoretically, we will show herein that the distribution of $D_k^*$ should approximate $F_{\mathbf{x}}$ as closely as possible. 

Before presenting the theoretical results, we introduce the necessary notations and conditions. For any $\eta$, $\xi\in\mathcal{L}_2(\mathcal{X})$, where $\mathcal{L}_2(\mathcal{X})$ is the set of square integrable functions on $\mathcal{X}$, let $\bar{V}(\eta,\xi)=\int_{\mathcal{X}}\eta(\mathbf{x})\xi(\mathbf{x})\mathbf{d}F_{\mathbf{x}}(\mathbf{x})$. We assess the prediction performance of $\tilde{f}_c$ over the whole region $\mathcal{X}$ with $V(\tilde{f}_{c}-f)\doteq \bar{V}(\tilde{f}_{c}-f,\tilde{f}_{ c}-f)=\int_{\mathcal{X}}\{\tilde{f}_{ c}(\mathbf{x})-f(\mathbf{x})\}^2\mathbf{d}F_{\mathbf{x}}(\mathbf{x})$, which is commonly employed in regression literature \citep{Ma2015Efficient,Rudi2015}, acknowledging the impact of $F_{\mathbf{x}}$ on the prediction.
Assuming $F_{\mathbf{x}}$ is a finite Borel measure on $\mathcal{X}$ and the covariance $c$ is continuous, Mercer's theorem \citep{Berlinet2004} suggests that $c(\mathbf{x}_i,\mathbf{x}_j)$ has an eigen expansion: $c(\mathbf{x}_i,\mathbf{x}_j)=\sum_{r=1}^{\infty}\rho_r^{-1}\phi_r(\mathbf{x}_i)\phi_r(\mathbf{x}_j)$ for   $\forall \mathbf{x}_i,\mathbf{x}_j\in \mathcal{X}$,
where the positive eigenvalues are in the decreasing order of $\rho_1^{-1}\ge \rho_2^{-1}\ge\cdots$, the eigen functions $\phi_r$ satisfy $\bar V(\phi_r,\phi_s)=\delta_{rs}$ and $\int_{\mathcal{X}}c(\mathbf{x}',\mathbf{x})\phi_r(\mathbf{x}')\mathbf{d}F_{\mathbf{x}}(\mathbf{x}')=\rho_r^{-1}\phi_r(\mathbf{x})$ with the Kronecker $\delta$. Then, any $\eta\in\mathcal{H}_c$ can be expanded as $\eta(\mathbf{x})=\sum_{r=1}^{\infty}\eta_r\phi_r(\mathbf{x})$ with $\eta_r=\bar V(\eta,\phi_r)$ and the norm $J_c(\eta)=\sum_{r=1}^{\infty}\rho_r\eta_r^2<\infty$. Three regularity conditions are required to analyze $V(\tilde{f}_c-f)$: 
[R1] $\mathcal{X}$ is a compact set with a Lipschitz boundary and satisfies an interior cone condition \citep{samplepath2018}. [R2] The eigenvalue sequence $\{\rho_r^{-1}\}_{r}$ follows a power law decrease. That is, $\rho_r>\beta r^{\gamma}$ holds for certain $\gamma>1$ and $\beta>0$. [R3] For any $r$ and $s$,  $V(\phi_r\phi_s)$ is bounded and the $(d+1)/2$-th order differential of $\phi_r\phi_s$ belongs to $\mathcal{L}_2(\mathcal{X})$.

Condition [R1] indicates that $\mathcal{X}$ is compact and has no $\prec$-shaped region on its boundary \citep{samplepath2018}. In [R2], the parameter $\gamma$ is important to indicate the smoothness of $Z$ and $c$. The smoother the process, the larger the value of $\gamma$. A more intuitive demonstration of $\gamma$ is reported in Section~\ref{sec:subsec:gammavalues}. [R3] shows two conditions about $\phi_r\phi_s$ for $\forall r,s$. The first one requires $V(\phi_r\phi_s)$ to be bounded, which virtually calls for a uniformly bounded fourth moment of $\phi_r(\mathbf{x})$ and is a widely-used regularity condition in the literature \citep{gu2013smoothing,Ma2015Efficient,meng2020}. The second one requires the $(d+1)/2$-th order differential of $\phi_r\phi_s$ to be square-integrable. It implies that $\phi_r\phi_s$ cannot be too rough. A similar condition has been used as a regularity condition in \citet{meng2020}. As \citet{gu2013smoothing} commented, the conditions are extremely difficult to verify because $\phi_r$ has no closed form.  However, they appear mild because $\phi_r$ typically decrease in smoothness but do not grow in magnitude. In Section S3 of the Supplementary Materials, we give an example satisfying [R1]--[R3] to help with the understanding of these conditions. Moreover, this example shows how $\gamma$ indicates the smoothness of $c$ or $Z$. 

Let the empirical distribution of $D_k^*$ on $\mathcal{X}$ be denoted as $F_k^*$ and let  $\mathbf{x}$, $\mathbf{x}'\stackrel{i.i.d.}{\sim}F_{\mathbf{x}}$. The energy distance between $F_k^*$ and $F_{\mathbf{x}}$ is $E(F_{\mathbf{x}},F_k^*)=\frac{2}{k}\sum_{j=1}^k\mathbb{E}(\|\mathbf{x}-\mathbf{x}_j^*\|_2)-\mathbb{E}(\|\mathbf{x}-\mathbf{x}'\|_2)-\frac{1}{k^2}\sum_{i=1}^k\sum_{j=1}^k\mathbb{E}(\|\mathbf{x}_i^*-\mathbf{x}_j^*\|_2)$.
This distance was proposed to test the goodness-of-fit of $F_k^*$ to $F_{\mathbf{x}}$ for high-dimensional $\mathbf{x}$ \citep{Energydistance2,Energydistance1,Supportpoints}, and it is used for the first time to assess the performance and influence of rep-points in the predictive process. Based on the link between GPs and KRRs, we examine the convergence rates of $\hat f_c$ and $\tilde{f}_c$. The following theorem indicates that $\tilde{f}_c$ based on rep-points can converge to $f$ at the same rate as $\hat f_c$ based on full data; this implication confirms the efficiency of low-rank approximations with rep-points. Moreover, the smoother the process, the faster is the convergence. 
\begin{Theorem}
(Convergence rate of $\tilde{f}_c$) Assume that the regularity conditions [R1]--[R3] hold, $\tau^2$ is a constant, and $\gamma>2$. As $n\to\infty$, the prediction in \eqref{eq:fullpred} has a convergence rate $V(\hat f_c-f)=O_p\left\{n^{-(1-1/\gamma)}\right\}$. For $D_k^*\subset \{\mathbf{x}_i\}_{i=1}^n$, if $n^{2/\gamma}E(F_{\mathbf{x}},F_k^*)\to 0$ as $k,n\to\infty$, the prediction $\tilde{f}_c$ in \eqref{eq:pp} achieves an identical convergence rate of $\hat f_c$
\vspace{-15pt}
\begin{equation}
    V(\tilde{f}_c-f)=O_p\left\{n^{-(1-1/\gamma)}\right\}.
    \label{eq:tildefc}
    \vspace{-15pt}
\end{equation}
\label{thm:tildefc}
\end{Theorem}
\vspace{-1cm}
Although Theorem~\ref{thm:tildefc} limits $D_k^*$ to a subset of $\{\mathbf{x}_i\}_{i=1}^n$, it guides the choice of the rep-points. The convergence of $\tilde{f}_c$ to $\hat f_c$ versus $k$ requires $E(F_{\mathbf{x}},F_k^*)=o_p(n^{-2/\gamma})$. This requirement indicates that the empirical distribution of rep-points $F_k^*$ should approach $F_{\mathbf{x}}$. Moreover, the faster the decrease in $E(F_{\mathbf{x}},F_k^*)$ with an increase in $k$, the earlier is the convergence of $\tilde{f}_c$ to $\hat f_c$ achieved. The abovementioned discussion inspired us to select $D_k^*$ that minimizes $E(F_{\mathbf{x}},F_k^*)$; this choice agrees with the definition of support points \citep[SPs;][]{Supportpoints}. The exact expression of $F_{\mathbf{x}}$ may be unknown in practice, and we can use the empirical distribution of $\{\mathbf{x}_i\}_{i=1}^n$, denoted as $F_n$, to replace $F_{\mathbf{x}}$. Using the R package \textit{support} \citep{supportRpackage}, we can easily obtain SPs. Note that SPs are generated rather than selected from $\{\mathbf{x}_i\}_{i=1}^n$, indicating that SPs may not belong to $\{\mathbf{x}_i\}_{i=1}^n$. 

Figures~\ref{fig:SP_Demo} and S2 show SPs, $k$-DPPs (obtained by an \href{https://github.com/camilacasquilho/k-dpp}{exact algorithm} based on \cite{kulesza2012determinantal}), random subsamples (Rands), and grid points (Grids) for multiple location sets $\{\mathbf{x}_i\}_{i=1}^n$. A comparison of Fig.~\ref{fig:subfig:SP_Unif} with \ref{fig:subfig:Rand_Unif} and~S2(a) demonstrates that for uniformly distributed $\{\mathbf{x}_i\}_{i=1}^n$, SPs are evenly spaced and cover the whole region well. A comparison of Figs.~\ref{fig:subfig:SP_NUnif} with \ref{fig:subfig:Grid_NUnif} demonstrates that for nonuniformly distributed $\{\mathbf{x}_i\}_{i=1}^n$, more SPs gather at the data-dense region, and sufficient SPs capture the data-sparse region as well. SPs have a space-filling property and mimic the location set of the full data efficiently, thereby maximizing the representativeness of each point \citep{Supportpoints}. This is another reason for terming $D_k^*$ as rep-points. 
\begin{figure}[h]
\centering
\subfigure[SPs]{
\label{fig:subfig:SP_Unif} 
\includegraphics[scale=0.35]{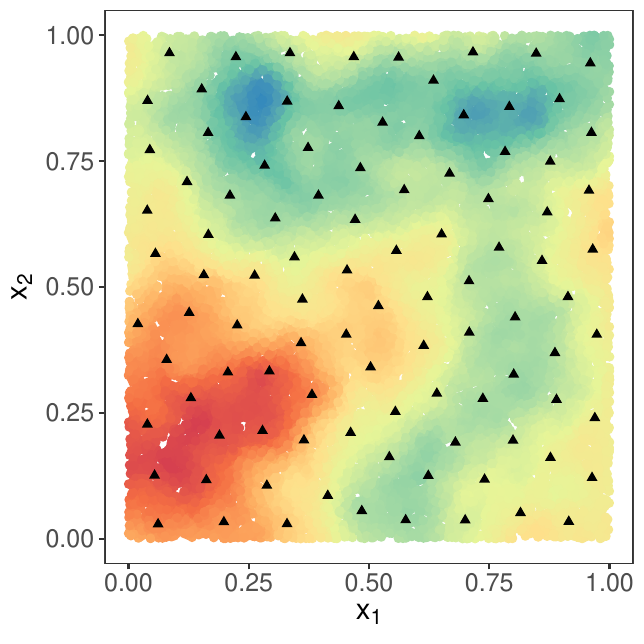}}
\subfigure[Rands]{
\label{fig:subfig:Rand_Unif} 
\includegraphics[scale=0.35]{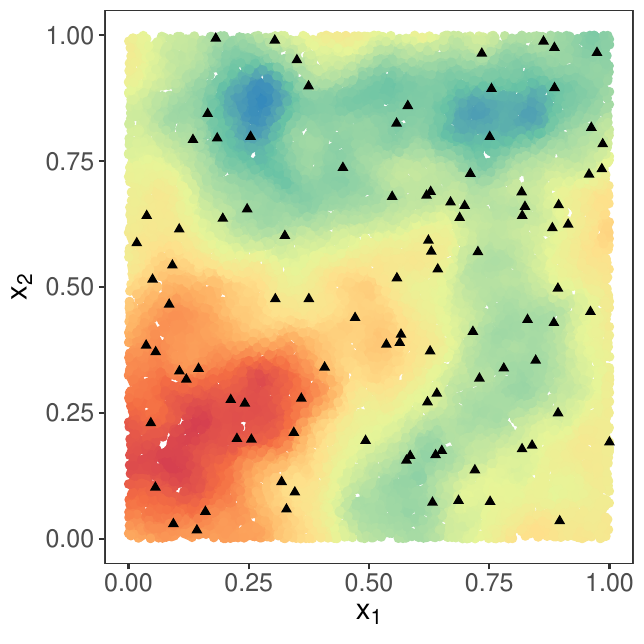}}
\subfigure[SPs]{
\label{fig:subfig:SP_NUnif} 
\includegraphics[scale=0.35]{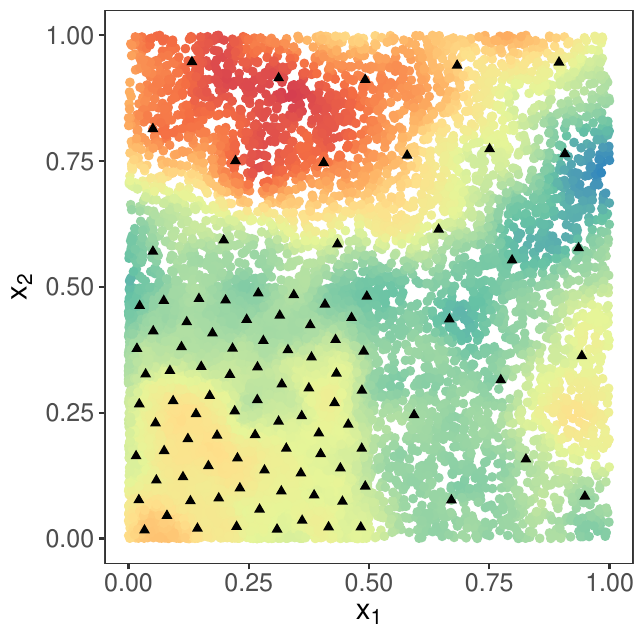}}
\subfigure[Grids]{
\label{fig:subfig:Grid_NUnif} 
\includegraphics[scale=0.35]{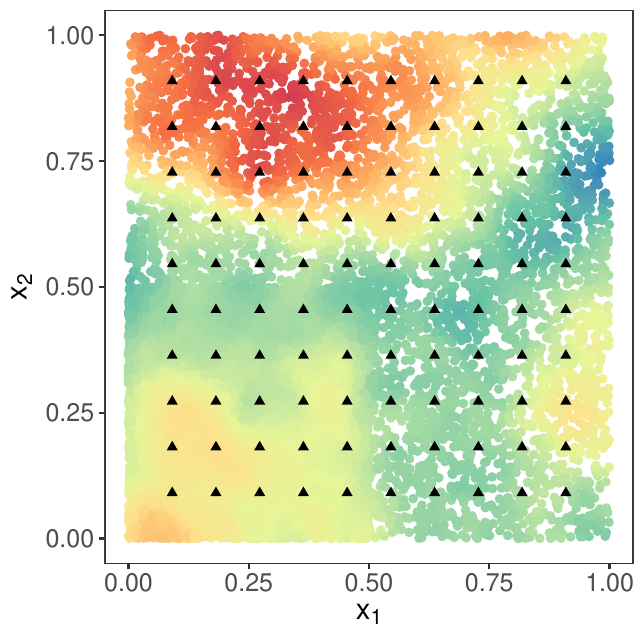}}
\caption{Rep-points for spatial data with $n=20000$ and multiple location sets. In (a) and (b), the location sets $\{\mathbf{x}_i\}_{i=1}^{n}$ follow a uniform distribution on $[0,1]^2$. In (c) and (d), $75\%$ of $\{\mathbf{x}_i\}_{i=1}^n$ follow a uniform distribution on $[0,0.5]^2$, and $25\%$ of $\{\mathbf{x}_i\}_{i=1}^n$ are uniformly distributed over the remaining region.}
\vspace{-10pt}
\label{fig:SP_Demo} 
\end{figure}

Finally, we discuss the order of $k$ required for achieving the convergence of $\tilde{f}_c$. \citet{Supportpoints} reported the relationship between the energy distance of SPs and $k$. For a bounded Borel set $\mathcal{X}$ with a nonempty interior,  $E(F_{\mathbf{x}},F_k^*)=O\{k^{-1}(\log k)^{-(1-\omega)/(2d)}\}$, where $\omega$ is in $(0,1)$. We consider $E(F_{\mathbf{x}},F_k^*)=O(k^{-1})$ as $k^{-1}$ dominates $(\log k)^{-(1-\omega)/(2d)}$. Then, to satisfy the condition of convergence, i.e., $E(F_{\mathbf{x}},F_k^*)=o_p(n^{-2/\gamma})$, we require that $k$ should have a higher order than $O(n^{2/\gamma})$. We can see that the larger the $\gamma$, the smoother is the process, and the smaller is the order of $k$ needed to ensure the convergence of $\tilde{f}_c$. 

Note that \citet{burt2019rates} and \citet{burt2020convergence} also derived the order of $k$ needed to make the KL divergence between the approximated model and the exact posterior to be arbitrarily small when inducing points are obtained via an approximated $k$-DPP and RLS. Their results share some common features with ours, although obtained via different theoretical frameworks. First, both derivations depend on the eigen decomposition of the covariance $c$ with  distribution $F_{\mathbf{x}}$. Therefore, the derived order of $k$ is determined by the smoothness of the process and the characteristics of $F_{\mathbf{x}}$. Second, both results suggest well dispersed rep-points over the domain $\mathcal{X}$. The $k$-DPP induces strong negative coefficients between the selected points. The support points suggested by us are space-filling. A detailed comparison of their performance as rep-points is presented in Section~\ref{sec:subsec:testrep}, highlighting the efficiency of SPs.


\subsection{Estimated covariance function}
Usually, the covariance $c$ is unknown and is replaced by an estimated one, i.e., $\hat c$. This section examines the influence of the estimated covariance $\hat c$ on the convergence rate of the original process and the conditions on rep-points for the predictive process such that its convergence can achieve the same rate. Similar to the studies by \citet{krigbound2} and \citet{krigbound3}, we demonstrate the undersmoothed and oversmoothed cases, where $\hat \gamma$ is smaller and larger than $\gamma$, respectively. Notation $a_n\preceq b_n$ for two positive sequences $a_n$ and $b_n$ indicates that there is a $B>0$ such that $a_n\le Bb_n$.

\begin{Theorem}(Convergence rate of $\tilde{f}_{\hat c}$) Assume that the regularity conditions [R1]--[R3] hold for $\hat c$ and $\hat \gamma, \gamma>2$. As $n\to\infty$, $V(\hat f_{\hat c}-f)=O_p\left\{n^{-\{\min(\hat\gamma,\gamma)-1\}/\hat\gamma}\right\}$. Furthermore, assume that $J_{\hat c}(f)\preceq O_p\{n^{-(\gamma-1-\hat\gamma)/\hat\gamma}\}$ when $\hat\gamma>\gamma-1$. Then, for $D_k^*\subset\{\mathbf{x}_i\}_{i=1}^n$, if $n^{2/\hat\gamma}E(F_{\mathbf{x}},F_k^*)\to 0$ as $k,n\to\infty$, the prediction $\tilde{f}_{\hat c}$ achieves a convergence rate identical to that achieved by $\hat f_{\hat c}$
\vspace{-15pt}
\begin{equation}
    V(\tilde{f}_{\hat c}-f)=O_p\left\{n^{-\frac{\min(\hat\gamma,\gamma)-1}{\hat\gamma}}\right\}.
    \label{eq:mis}
    \vspace{-15pt}
\end{equation}
\label{thm:mis}
\end{Theorem}
\vspace{-30pt}

Theorem~\ref{thm:mis} shows the convergence rates of both $\hat f_{\hat c}$ and $\tilde{f}_{\hat c}$. For the undersmoothed case, i.e., $\hat\gamma<\gamma$, the convergence rate is $O_p\{n^{-(1-1/\hat\gamma)}\}$. For the oversmoothed case, i.e., $\hat\gamma>\gamma$, the convergence rate becomes $O_p\{n^{-(\gamma-1)/\hat\gamma}\}$. The closer $\hat\gamma$ is to $\gamma$, the faster is the convergence. When $\hat\gamma=\gamma$, Theorem~\ref{thm:mis} becomes Theorem~\ref{thm:tildefc}. Even in the case of an estimated covariance, the convergence rate of $\tilde{f}_{\hat c}$ can be identical to that of $\hat f_{\hat c}$ with sufficient rep-points. To achieve convergence, we require $E(F_{\mathbf{x}},F_k^*)=o_p\{n^{-2/\hat\gamma}\}$, indicating that $\tilde{f}_{\hat c}$ in the undersmoothed case requires additional rep-points than $\tilde{f}_c$ to achieve convergence; in the oversmoothed case,  $\tilde{f}_{\hat c}$ achieves convergence faster. 

With an estimated covariance, the following theorem suggests the adjustment of $\tau^2$ to improve the convergence rates of $\hat f_{\hat c}$ and $\tilde{f}_{\hat c}$ when $k$ is given and large enough.
\begin{Theorem}(Adjusted $\tau^2$) Assume that the regularity conditions [R1]--[R3] hold for $\hat c$ and $\hat \gamma,\gamma>2$. If $\tau^2=O(n^{-\hat\gamma/\gamma+1})$, the prediction based on full data achieves the optimal convergence rate, i.e.,  $V(\hat f_{\hat c}-f)=O_p\left\{n^{-(1-1/\gamma)}\right\}$ as $n\to\infty$. Furthermore, assume that $J_{\hat c}(f)\preceq O_p\{n^{(\gamma-1-\hat\gamma)/\gamma}\}$ when $\hat\gamma>\gamma-1$. Then, for  $D_k^*\subset\{\mathbf{x}_i\}_{i=1}^n$, if  $n^{2/\hat\gamma}E(F_{\mathbf{x}},F_k^*)\to 0$ as $k,n\to\infty$, the prediction based on the rep-points achieves the optimal convergence rate
\vspace{-20pt}
\begin{equation*}
   V(\hat f_{\hat c}-f)=O_p\left\{n^{-(1-1/\gamma)}\right\}.
\vspace{-20pt}
\end{equation*}
\label{thm:adjuestedtau}
\end{Theorem}
\vspace{-35pt}

Theorem~\ref{thm:adjuestedtau} demonstrates that even with a mis-specified covariance, $\tilde{f}_{\hat c}$ and $\hat f_{\hat c}$ can converge with the same rate as $\tilde{f}_{c}$ and $\hat f_{c}$, respectively, by properly adjusting the order of $\tau^2$. For the undersmoothed case, the adjusted $\tau^2$ has a higher order than $O(1)$, and vice versa. The larger the $|\hat \gamma-\gamma|$, the larger is the required adjustment of $\tau^2$. For $\hat\gamma=\gamma$, we do not require to adjust $\tau^2$ and Theorem~\ref{thm:adjuestedtau} becomes Theorem~\ref{thm:tildefc}. Although we can never adjust $\tau^2$ to be of the exact order desired, Theorem~\ref{thm:adjuestedtau} provides certain guidance. For the example of the undersmoothed case, we can make $\tau^2$ larger than the initial value to improve the prediction performance of $\tilde{f}_{\hat c}$ and $\hat f_{\hat c}$.

\section{Simulations}
\label{sec:simulation}
This section uses simulated data to show the properties of the smoothness parameter $\gamma$, demonstrate the advantages of low-rank approximation with SPs under various scenarios, and confirm our theoretical results. The computations were implemented by R 3.6.3 \citep{Rcite} on a machine with Intel(R) Xeon(R) CPU E5-2680 v4 @ 2.40GHz and 125 GB RAM.

\subsection{Smoothness parameter $\gamma$}
\label{sec:subsec:gammavalues}
We use the Mat\'{e}rn covariance function as an example to provide an intuitive demonstration of the smoothness parameter $\gamma$. Assume that the process is stationary and $c$ has the following form:
\vspace{-25pt}
\begin{equation}
    c(\mathbf{x}_i,\mathbf{x}_j)=\frac{\sigma^2}{\Gamma(\nu)2^{\nu-1}}\left(\frac{\|\mathbf{x}_i-\mathbf{x}_j\|_2}{\psi}\right)^{\nu}\mathcal{K}_{\nu}\left (\frac{\|\mathbf{x}_i-\mathbf{x}_j\|_2}{\psi}\right),
    \label{eq:matern}
    \vspace{-10pt}
\end{equation}
where $\sigma^2>0$ is the marginal variance such that $c(\mathbf{x},\mathbf{x})=\sigma^2$, $\psi>0$ is called the range parameter; $\nu>0$ determines the smoothness of $f(\mathbf{x}_i)$ (or $Z$) and is termed the smoothness parameter in Mat\'{e}rn; $\mathcal{K}_{\nu}$ is the modified Bessel function of the second type of order $\nu$. 
\begin{figure}[b!]
\centering
\subfigure[]{
\label{fig:subfig:Gamma_nu} 
\includegraphics[scale=0.45]{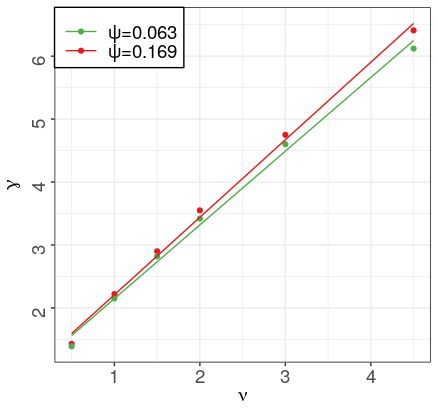}}
\hspace{5mm}
\subfigure[]{
\label{fig:subfig:Gamma_phi} 
\includegraphics[scale=0.45]{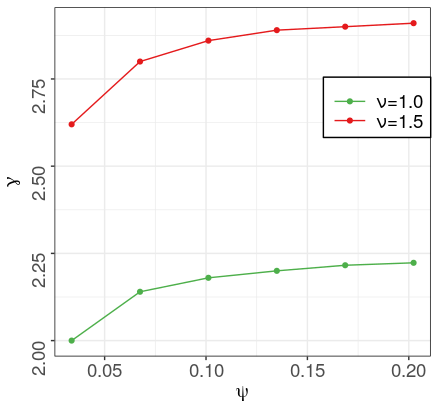}}
\caption{Values of the smoothness parameter $\gamma$ for the Mat\'{e}rn covariance function  \eqref{eq:matern} under various parameter settings: (a) weak and strong correlations; (b) rough and smooth realizations.}
\label{fig:Gammavalues} 
\vskip -10pt
\end{figure}

For obtaining an empirical estimate of $\gamma$, we calculate the eigenvalues of an $n_0\times n_0$ matrix $\mathbf{C}_0=\{c(\mathbf{x}_i,\mathbf{x}_j)\}_{i,j=1}^{n_0}$, i.e., $\lambda_{01}, \ldots, \lambda_{0n_0}$, where $\{\mathbf{x}_i\}_{i=1}^{n_0}$ are uniformly distributed on $[0,1]^2$. As $\lambda_{0i}^{-1}=O(i^{\gamma})$, we fit $\{(\log i,-\log \lambda_{0i})\}_{i=1}^{n_0}$ with a linear regression and obtain the slope as an estimate of $\gamma$. Fig.~\ref{fig:Gammavalues} shows how $\gamma$ changes with the parameters in the Mat\'{e}rn covariance, where $\sigma^2$ is fixed at $1.5$. In Fig.~\ref{fig:subfig:Gamma_nu}, we set $\psi$ to be $0.063$ and $0.169$ for covering the cases of weak and strong correlations, respectively. Based on other parameters, the smoothness parameter $\gamma$ linearly grows with the parameter $\nu$; this result is in agreement with the discussion reported by \citet{krigbound3} that $\gamma$ is about $2(\nu+d/2)/d$. In Fig.~\ref{fig:subfig:Gamma_phi}, $\gamma$ increases with $\psi$, though not linearly. Moreover, although $\psi$ varies within a large range, $\gamma$ varies within $0.5$. Compared with $\psi$, $\gamma$ is more sensitive to the value of $\nu$.

Two Mat\'{e}rn covariances with consistent parameters $(\sigma_1^2,\psi_1,\nu)$ and $(\sigma_2^2,\psi_2,\nu)$ satisfying $\sigma_1^2/\psi_1^{2\nu}=\sigma_2^2/\psi_2^{2\nu}$ define two equivalent probability measures and thus obtain asymptotically equal predictions \citep{zhang2004}. Our simulation and theoretical results taken together demonstrate that the prediction performance is determined by the $\gamma$ value, and two consistent parameter settings should have an asymptotically equal $\gamma$. Table~\ref{tab:GammaConsistent} lists the average values of $\gamma$ among ten replicates under three consistent parameter settings. The average values get closer as $n_0$ increases. This result supports our claim that $\gamma$ is a crucial measure of the smoothness of the process. Further validation is provided through an example utilizing another covariance, as demonstrated in Section S3 of the Supplementary Materials.

\begin{table}[h]
\centering
\caption{Values of $\gamma$ under consistent parameter settings obtained from the data of increasing size $n_0$ of the Mat\'{e}rn covariance function \eqref{eq:matern}.}
\vspace{5pt}
{
	\begin{tabular}{l|cccc}
			\toprule
			$(\sigma^2,\psi,\nu)$ & $n_0=2000$ & $4000$ & $6000$ & $8000$ \\ 
			\midrule
			{$(1.5,0.169,1.5)$} & $2.8899$ & $2.9011$ & $2.9064$ & $2.9026$ \\ \hline
			{$(1.0,0.147,1.5)$} & $2.8666$ & $2.8865$ & $2.8954$ & $2.8937$ \\ \hline
			{$(2.0,0.186,1.5)$} & $2.9045$ & $2.9102$ & $2.9132$ & $2.9082$ \\ 
			\bottomrule	
	\end{tabular}}
 \vspace{-10pt}
\label{tab:GammaConsistent}
\end{table}

\subsection{Performance of various rep-points choices}
\label{sec:subsec:testrep}
The performance of various rep-points choices under the following four scenarios were tested and compared: 
\begin{enumerate}
    \vspace{-5pt}
    \item \textbf{Strong correlation}: $Z\sim\mathrm{GP}(0,c_s)$, where $c_s$ belongs to the Mat\'{e}rn class in \eqref{eq:matern} with $\bm{\theta}_s=(1.5, 0.169, 1.5)^\top$. The effective range of $Z$ is $0.8$. Here $F_{\mathbf{x}}$ is the uniform distribution on $[0,1]^2$, i.e., $U_{[0, 1]^2}$. Four choices of rep-points (SPs, $k$-DPPs, Grids and Rands) of size $k=(0.10, 0.18, 0.28, 0.40, 0.55, 0.81, 1.12, 1.36)\times n^{2/2.9}\approx (36, 64, 100, 144, 196,$ $289, 400, 484)$ were selected based on the full data, where $2.9$ is the estimated $\gamma$ for $c_s$ obtained from Table~\ref{tab:GammaConsistent}. Two mis-specified covariances with $\hat{\bm{\theta}}_s=(1.5, 0.169, 1.0)^\top$ and $(1.5,0.169, 3.0)^\top$ were imposed as the undersmoothed and oversmoothed cases, respectively.
    \vspace{-5pt} 
    \item \textbf{Weak correlation}: $Z\sim\mathrm{GP}(0,c_w)$, where $c_w$ belongs to the Mat\'{e}rn class in \eqref{eq:matern} with $\bm{\theta}_w=(1.5, 0.063, 1.5)^\top$. The effective range of $Z$ is $0.3$, and $F_{\mathbf{x}}=U_{[0, 1]^2}$. SPs, $k$-DPPs, Grids and Rands of sizes $k=(0.66, 1.21, 1.66, 2.05, 2.63, 3.12, 3.64, 4.02)\times n^{2/2.8}\approx (289, 529, 729, 900, 1156$, $1369, 1600, 1764)$ are chosen, where $2.8$ is the estimate of $\gamma$ for $c_w$ obtained from Fig.~\ref{fig:subfig:Gamma_nu}. Two mis-specified covariances with $\hat{\bm{\theta}}_w=(1.5, 0.063, 1.0)^\top$ and $(1.5,0.063, 3.0)^\top$ were imposed as undersmoothed and oversmoothed cases, respectively.
    \vspace{-5pt} 
    \item \textbf{Consistency}: The settings are the same as those in the first scenario, although parameters of mis-specifed covariances are substituted as $\hat{\bm{\theta}}_{s}=(1.0, 0.147, 1.5)^\top$ and $(2.0, 0.186, 1.5)^\top$ to be consistent with the parameters of $c_s$ as discussed in  Section~\ref{sec:subsec:gammavalues}. 
    \vspace{-30pt} 
    \item \textbf{Nonuniform $F_{\mathbf{x}}$}: $Z\sim\mathrm{GP}(0,c_s)$. The covariance $c_s$ is the same as in the first scenario. Here, we have $75\%$ of $\{\mathbf{x}_i\}_{i=1}^n$ uniformly distributed in $[0.0,0.5]^2$ and the other $25\%$ uniformly distributed in the remaining region of $[0,1]^2$, as described in Fig.~\ref{fig:subfig:SP_NUnif} and \ref{fig:subfig:Grid_NUnif}. Five choices of rep-points were considered, i.e., Grids, $k$-DPPs, Rands and SPs generated by $\{\mathbf{x}_i\}_{i=1}^n$ (SPs) and by $U_{[0, 1]^2}$ (SPUs). The size $k$ was the same as that in the first scenario.
\end{enumerate}
For each scenario, we generate $\{(\mathbf{x}_i,f_i)\}_{i=1}^{n+n_t}$, where $f_i=f(\mathbf{x}_i)$'s are the realizations of the GP $Z$ and $n=n_t=5000$. Then, we randomly divide the data into two parts of sizes $n$ and $n_t$. Next, we perturb the training data with random noises from $\mathcal{N}(0,0.27)$. Finally, we evaluate $f(\mathbf{x})$ on the remaining $n_t$ locations $\{\mathbf{x}_{ti}\}_{i=1}^{n_t}$ based on the rep-points or full data. The performance of each method is evaluated in terms of the root mean squared prediction error (RMSPE), i.e., $[n_t^{-1}\sum_{i=1}^{n_t}\{\tilde{f}_{\hat c}(\mathbf{x}_{ti})-f(\mathbf{x}_{ti})\}^2]^{1/2}$. We show the average of RMSPEs over 100 replicates of the abovementioned process in Fig.~\ref{fig:VersusK}. 

In Figs.~\ref{fig:subfig:Set1} and \ref{fig:subfig:Set2}, the RMSPEs for $\tilde{f}_c$ and $\tilde{f}_{\hat c}$ decrease at first and then converge to the levels of $\hat{f}_c$ and $\hat{f}_{\hat c}$, respectively, as $k$ increases, regardless of the types of rep-points used. This result agrees with our theoretical results that $\tilde{f}_{\hat c}$ and $\hat f_{\hat c}$ may achieve the same convergence rate if $D_k^*$ (or rep-points) approximate $F_{\mathbf{x}}$ well enough and validates the efficiency of the rep-points. The solid lines representing SPs exhibit the fastest decline, regardless of whether the imposed/estimated covariance is correct or not. This result highlights the superior efficiency of SPs in representing the set $\{\mathbf{x}_i\}_{i=1}^n$. The performance of $k$-DPPs falls between that of SPs and Grids. Note that the $k$-DPP algorithm becomes unavailable as $\nu$ and $k$ increase. When the covariance is smoother, $k$-DPPs tend to be more dispersed. However, placing an excessive number of $k$-DPPs within a fixed domain would disrupt the negative correlations between them. Rands exhibit the worst performance because of the lack of the space-filling property. 

\begin{figure}[!h]
\centering
\subfigure[Strong correlation]{
\label{fig:subfig:Set1} 
\includegraphics[scale=0.5]{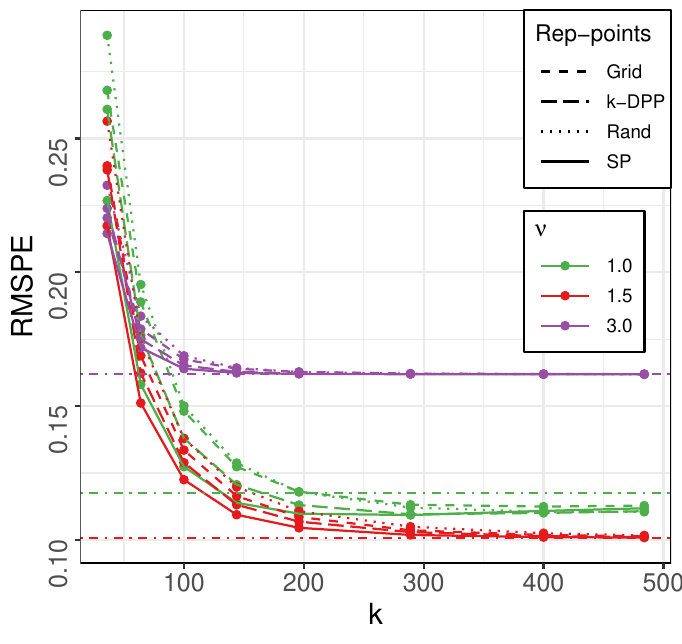}}
\hspace{5mm}
\subfigure[Weak correlation]{
\label{fig:subfig:Set2} 
\includegraphics[scale=0.5]{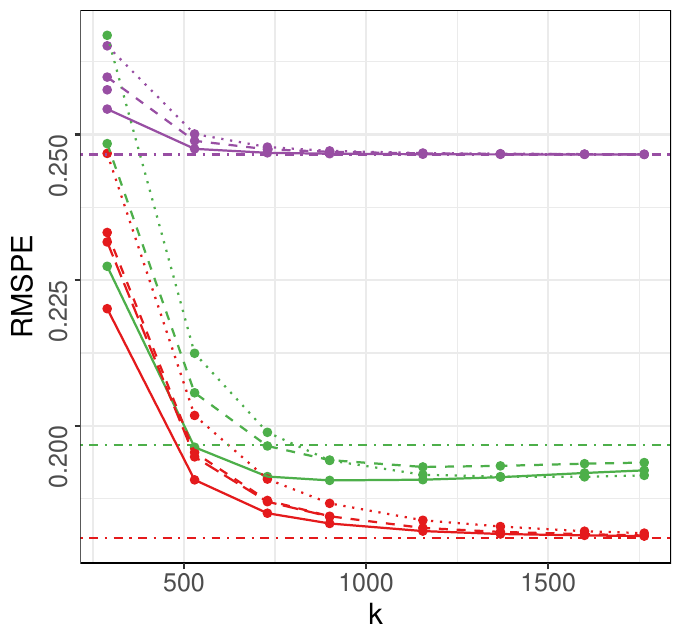}}\\
\subfigure[Consistent]{
\label{fig:subfig:Set3} 
\includegraphics[scale=0.5]{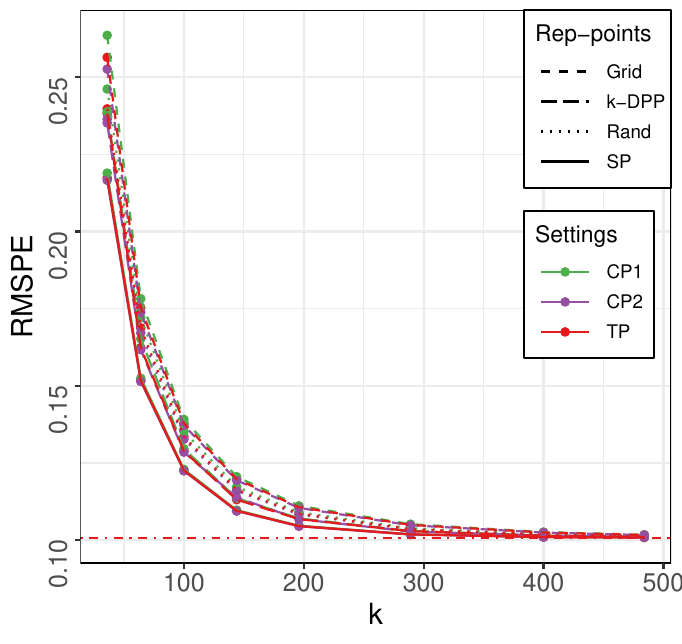}}
\hspace{5mm}
\subfigure[Nonuniform $F_{\mathbf{x}}$]{
\label{fig:subfig:Set4} 
\includegraphics[scale=0.5]{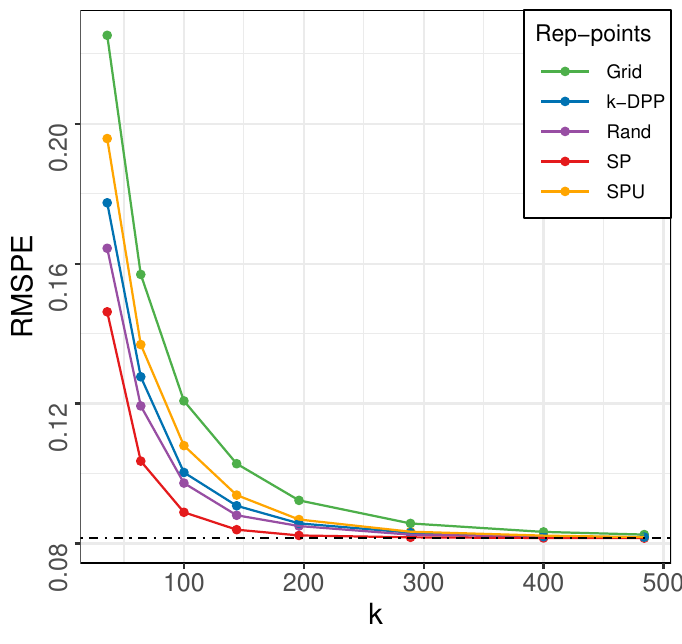}}
\caption{Comparison of various rep-points under four scenarios (a)--(d). The dashed--dotted lines represent the RMSPEs of predictions based on the full data. The legend in (b) is the same as that in (a). CP1, CP2, and TP in (c) correspond to the consistent parameters $\hat{\bm{\theta}}_{s1}$, $\hat{\bm{\theta}}_{s2}$, and true parameter $\bm{\theta}_s$, respectively, in the third setting.}
\label{fig:VersusK} 
\end{figure}

To understand the influence of the imposed covariance, we compare curves of the same type but shown in different colors. When $k$ is not extremely small, the predictions given by the true covariance have the lowest RMSPEs, indicating that a wrongly specified covariance can typically weaken the prediction performance. When $k$ is quite small, imposing an oversmoothed covariance function may yield a more accurate prediction. A larger value of $\hat\gamma$ (or $\hat\nu$) results in the earlier convergence of the corresponding $\tilde f_{\hat c}$; this agrees with our theoretical results, according to which the required $k$ for the convergence of $\tilde{f}_{\hat c}$ is of $O(n^{2/\hat\gamma})$. 

The predictions shown in Fig.~\ref{fig:subfig:Set1} have smaller RMSPEs than those shown in Fig.~\ref{fig:subfig:Set2}. The second setting involves a more complex structure with a smaller $\psi$ and hence a smaller $\gamma$. Moreover, $\tilde{f}_{\hat c}$ and $\tilde{f}_{c}$ in Fig.~\ref{fig:subfig:Set2} require a larger $k$ to achieve convergence, further demonstrating the influence of shortening of the range parameter $\psi$. In particular, define $K$ as the size of rep-points such that $\tilde{f}_c$ based on SPs achieves convergence. Remember that $K$ has a higher order than $O(n^{2/\gamma})$. Then, both $\nu$ and $\psi$ influence $K$. When $\psi$ is fixed, $K$ changes less with $\nu$. For example, $K$ for different lines (in different colors) in Fig.~\ref{fig:subfig:Set1} are all within $500$. However, when we change the value of $\psi$, $K$ drastically varies. For example, the $K$ for the red solid line in Fig.~\ref{fig:subfig:Set2} is much larger than that shown in Fig.~\ref{fig:subfig:Set1}. 

Figure~\ref{fig:subfig:Set3} shows the performance of predictions with consistent parameters. Curves represented in different colors but of the same type are approximately coincident. That is, by controlling the choice of rep-points, we can achieve $\tilde{f}_{\hat c}$ with consistent $\hat c$ to give an asymptotically identical prediction performance. This result agrees with our discussion reported in Section~\ref{sec:subsec:gammavalues} that consistent covariance functions have asymptotic identical values of $\hat\gamma$, which determine the performance of $\tilde{f}_{\hat c}$. Moreover, we confirmed the conclusion of \citet{zhang2004}, i.e., an incorrect but consistent $\hat c$ has the same prediction performance as $c$. 

\begin{table}[!b]
\centering
\caption{$E(F_k^*,F_n)$ for multiple rep-points with $n=5000$. Furthermore, $k$ is varied in the fourth scenario. The minimum $E(F_k^*,F_n)$ for each $k$ is highlighted in bold.}
\vspace{2mm}
\footnotesize
{
	\begin{tabular}{l|cccccccc}
			\toprule
			Rep-points & $k=36$ & $64$ & $100$ & $144$ & $196$ & $289$ & $400$ & $484$ \\ 
			\midrule
             {SPs} & $\textbf{0.001081}$ & $\textbf{0.000455}$ & $\textbf{0.000237}$ & $\textbf{0.000141}$ & $\textbf{0.000093}$ & $\textbf{0.000056}$ & $\textbf{0.000039}$ & $\textbf{0.000034}$ \\ \hline
			{Rands} & $0.013513$ & $0.007517$ & $0.004582$ & $0.003171$ & $0.002268$ & $0.001555$ & $0.001244$ & $0.001053$ \\ \hline
            {$k$-DPP} & $0.017491$ & $0.014829$ & $0.013808$ & $0.012821$ & $0.012957$ & $0.012167$ & $0.011523$ & $0.011610$ \\ \hline
			{SPUs} & $0.030338$ & $0.029755$ & $0.029427$ & $0.029356$ & $0.029288$ & $0.029318$ & $0.029379$ & $0.029305$ \\ \hline
			{Grids} & $0.034862$ & $0.032329$ & $0.031250$ & $0.030650$ & $0.030270$ & $0.029863$ & $0.029729$ & $0.029670$ \\
			\bottomrule	
	\end{tabular}}
\label{tab:Energydistance}
\end{table}

Fig.~\ref{fig:subfig:Set4} shows the performance of $\tilde{f}_c$ based on various rep-points when $\{\mathbf{x}_i\}_{i=1}^n$ are nonuniformly distributed over a compact region $\mathcal{X}$. $\tilde{f}_c$ based on SPs, $k$-DPPs, and Rands converge earlier than those based on SPUs and Grids. This is because they follow $F_{\mathbf{x}}$, whereas SPUs and Grids follow $U_{[0,1]^2}$. In other words, the distribution of rep-points should be as close to $F_n$ as possible. SPs meet this requirement by minimizing the energy distance between their empirical distribution and $F_n$. For an intuitive comparison, we provide $E(F_k^*,F_n)$ for the five rep-points in Table~\ref{tab:Energydistance}. All the energy distances decrease as $k$ grows. SPUs and Grids have larger $E(F_k^*,F_n)$ than SPs and Rands because they follow a distribution different from $F_{\mathbf{x}}$. Even with the same asymptotic distribution, there exist large gaps between the energy distances of various rep-points, particularly for small $k$. For example when $k=484$, the $E(F_n,F_k^*)$ for SPs is only $3.23\%$ and $0.29\%$ of those for Rands and $k$-DPPs, respectively. SPs are the most efficient in minimizing $E(F_k^*,F_n)$. The performance of $k$-DPPs is affected by the non-uniform distribution. 

\subsection{Influence of $\hat c$}
\label{sec:subsec:3.3}
First, we confirm Theorem~\ref{thm:mis} by exploring the orders of $V(\tilde{f}_{\hat c}-f)$ and $V(\tilde{f}_c-f)$ with respect to $n$ for a sufficiently large $k$. Let $F_{\mathbf{x}}=U_{[0,1]^2}$, $n=1000, 1500, \ldots, 7000$, and $n_t=5000$. For each $n$, we follow the procedure described in Section~\ref{sec:subsec:testrep} to generate the training data and to predict on the testing location based on SPs and Rands with $k=\lfloor 1.5\times n^{2/2.9}\rfloor$. We consider $c_s$ in the first scenario as the true covariance such that $\tilde{f}_{\hat c}$ and $\tilde{f}_c$ can converge to $\hat f_{\hat c}$ and $\hat f_c$ with a smaller $k$. The values of $\sigma^2$ and $\psi$ in $c_s$ are fixed; however, $\nu$ is set to be $0.5$ and $1.0$ for undersmoothed covariances and $2.0$ and $3.0$ for oversmoothed ones. 

The MSPEs for $\tilde{f}_{\hat c}$ and $\tilde{f}_c$ may well approximate $V(\tilde{f}_{\hat c}-f)$ and $V(\tilde{f}_{c}-f)$, respectively, as per the law of large numbers and because $n_t$ is quite large. Let MSPE($n$) be the MSPE obtained from the training data of size $n$. Then, by Theorem~\ref{thm:mis}, we have $\log\{\mathrm{MSPE}(n)\}\approx -\frac{\min(\gamma,\hat\gamma)-1}{\hat \gamma}\log n+\log C_0$,
where $C_0$ is the multiplicative constant of the order term. Our target is to fit the linear regressions above and analyze the slopes. 

\begin{figure}[!t]
\centering
\includegraphics[scale=0.5]{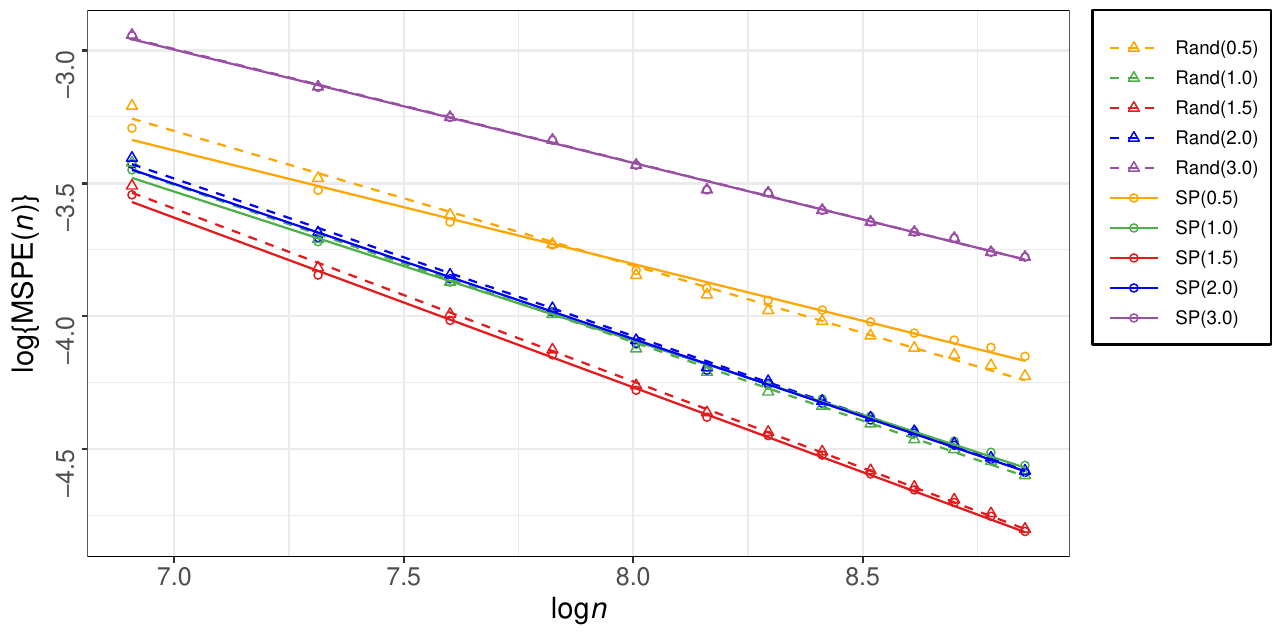}
\caption{Prediction performance of $\tilde{f}_{\hat c}$ and $\tilde{f}_c$ with varying $n$ and a sufficiently large $k$. Rand$(0.5)$ in the legend indicates the performance of $\tilde{f}_{\hat c}$ with Rands and $\nu=0.5$, and so on for the rest.}
\label{fig:slopes}
\end{figure} 
In Fig.~\ref{fig:slopes}, we show the average $\log\{\mathrm {MSPE}(n)\}$s of $100$ replicates (using points) and their fitted lines. For simplicity of discussion, only SPs and Rands are considered. The points for each $n$ and lines in the same color are nearly coincident in most cases, indicating that $k=\lfloor 1.5\times n^{2/2.9}\rfloor$ is enough for both rep-points in this scenario, and hence $\tilde{f}_{\hat c}$ and $\tilde{f}_c$ have converged to $\hat f_{\hat c}$ and $\hat f_c$, respectively. Red lines have the largest negative slopes (NSs) because the prediction under the true covariance has the best performance. Additionally, the value of $C_0$ for red lines is around $2.4$. For undersmoothed cases, the NSs for SP$(0.5)$ and Rand$(0.5)$ are smaller than those for SP$(1.0)$ and Rand$(1.0)$. For oversmoothed cases, the NSs for SP$(3.0)$ and Rand$(3.0)$ are smaller than those for SP$(2.0)$ and Rand$(2.0)$. The worse the estimate of $\hat c$, the larger is the $|\hat\gamma-\gamma|$; hence, the worse is the behavior of $\tilde{f}_{\hat c}$. 

\begin{figure}[b]
\centering
\includegraphics[scale=0.5]{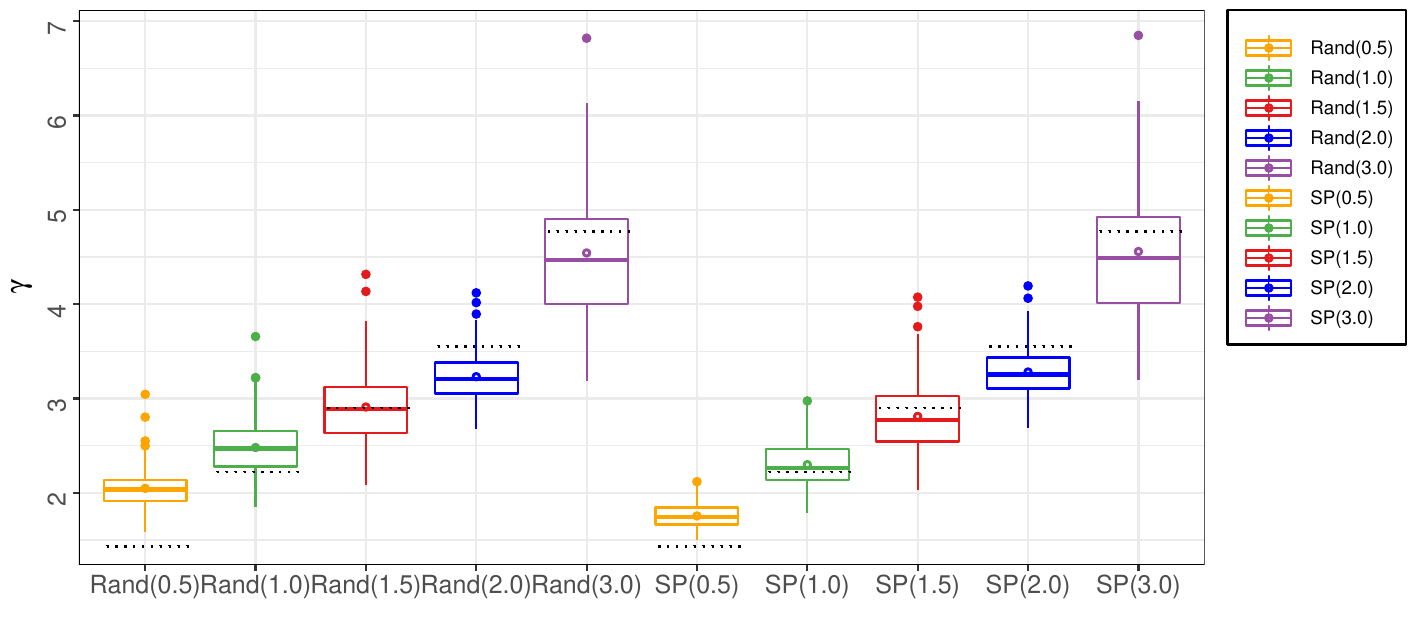}
\caption{Boxplots for the theoretically obtained  $\hat\gamma$ and $\gamma$. The circles in the middle of boxes are the average values. The black dotted lines are $\hat\gamma$ and $\gamma$, as reported in Section~\ref{sec:subsec:gammavalues}.}
\label{fig:VersusN}
\end{figure} 

As per \eqref{eq:tildefc} and \eqref{eq:mis}, we can evaluate the values of $\hat\gamma$ and $\gamma$ through NS. That is, $\hat\gamma\approx 1/(1-\mathrm{NS})$ for undersmoothed cases and $\hat\gamma\approx (\gamma-1)/(1-\mathrm{NS})$ for oversmoothed cases. The results from $100$ replicates are summarized in Fig.~\ref{fig:VersusN}. For comparison, we provide the empirical estimates of $\hat\gamma$ and $\gamma$ obtained as described in Section~\ref{sec:subsec:gammavalues}. In most cases, boxes cover the respective black dotted lines well, thereby confirming our theory. For the undersmoothed covariance with $\nu=0.5$, the black dotted lines are considerably lower than their boxes because their smoothness parameters are smaller than $2$. Hence our theorems do not apply in these cases. Our theory was further confirmed by setting $\sigma^2$ as $0.5$ while maintaining the other parameters fixed for both the true and mis-specified covariances. This verification is illustrated in Fig.~S1.

Finally, we demonstrate the improvement of $\tilde{f}_{\hat c}$ by adjusting $\tau^2$. Still using the first scenario and $\tau^2=0.27$, we generate $\{(\mathbf{x}_i,y_i)\}_{i=1}^n$ with $n=5000$ and $\{(\mathbf{x}_{ti},f_{ti})\}_{i=1}^{n_t}$ with $n_t=5000$ in the same way. The predictions $\tilde{f}_c(\mathbf{x}_{ti})$ and $\tilde{f}_{\hat c}(\mathbf{x}_{ti})$ are calculated using the adjusted $\tau^2=10^{(-3, 2, 1, 0, 0.5, 1, 1.5, 2)}\times 0.27$.  Fig.~\ref{fig:TAU} summarizes the average $\log(\rm MSPE)$s over $100$ replicates versus $a_{\tau}=\log_{10}(\tau^2/0.27)$. When the correct covariance ($\nu=1.5$) is used, the red curve achieves the lowest MSPE when $a_{\tau}=0$, and this lowest MSPE corresponds to the true value of $\tau^2$. When the covariance is mis-specified ($\nu=1$ or $3$), neither of the two curves reach the minimum MSPE under the true $\tau^2$. However, the performance of the oversmoothed covariance (with $\nu=3$) improves when $a_{\tau}$ is reduced to $-2$; the performance of the undersmoothed covariance (with $\nu=1$) improves when $a_{\tau}$ is increased to $0.5$. Furthermore, the lowest MSPEs of $\tilde{f}_{\hat c}$ are comparable to those of $\tilde{f}_c$; this result is in agreement with Theorem~\ref{thm:adjuestedtau} that $\tilde{f}_{\hat c}$ may have the same convergence rate as $\tilde{f}_c$ for an appropriate value of $\tau^2$.
\begin{figure}[!h]
\centering
\includegraphics[scale=0.5]{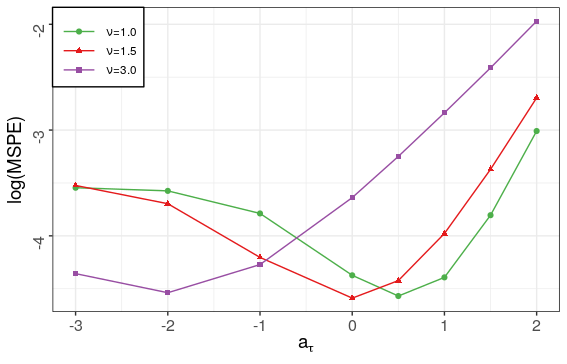}
\caption{Prediction performance of $\tilde{f}_{\hat c}$ and $\tilde{f}_c$ with varying $\tau^2$.}
\label{fig:TAU}
\end{figure} 
\vspace{-.7cm}

\section{Data Examples}
\label{sec:realdata}
This section shows the advantages of low-rank approximations with an estimated covariance and SPs on two real datasets: annual total precipitation anomalies and total column ozone data. We compare our method with other popular methods such as FRK \citep{FRK_intro}, LatticeKrig \citep{LKPackage}, and autoFRK \citep{autoFRKPackage}.

\subsection{Annual total precipitation anomalies data}
In the first application, we consider the data of annual total precipitation anomalies observed at $7352$ weather stations in the United States in 1962. This dataset was discussed by \citet{taper2008} and \citet{sang2013}. As shown in Fig.~S3, the locations are nonuniformly distributed over an irregular region $\mathcal{X}$. Without an obvious mean trend, nonstationarity or anisotropy \citep{taper2008}, we fit the data with a stationary Gaussian process with mean zero. In particular, we randomly divide the data into training data with $n=7000$ and testing data with $n_t=352$. Based on the training data, we apply the Vecchia approximation to obtain an approximated covariance estimation. This approximation can be implemented by the R package \textit{GpGp} \citep{GpGppackage}, where we set the covariance to be Mat\'{e}rn in \eqref{eq:matern}. Furthermore, we use \textit{ExaGeoStat} \citep{ExaGeoStat1} to calculate the exact covariance estimation as a benchmark. Using these approximated covariance and SPs of $k=210$, $500$, $750$ and $1000$, we obtain $\tilde{f}_{\hat c}(\mathbf{x}_{ti})$s. We calculate the $\mathrm{MSPE}=n_t^{-1}\sum_{i=1}^{n_t}\{y_{ti}-\tilde{f}_{\hat c}(\mathbf{x}_{ti})\}^2$ to measure the prediction performance of our method, where $\{y_{ti}\}_{i=1}^{n_t}$ are the responses of the testing data. 

Moreover, we consider the method of using the estimated covariance and full data (FULL) and three popular low-rank approximation methods mentioned earlier, namely, FRK, LatticeKrig, and autoFRK. The FRK performs inference on basic areal units (BAUs). It develops multiresolution basis functions and calculates coefficients by using the expectation-maximization (EM) algorithm. We set two and three levels of resolutions and exponential kernel functions. LatticeKrig induces the basis functions with regular lattices of different levels and uses the Markov random field assumption to develop a sparse precision matrix for random weights. The higher the level, the finer are the grid points. This method can be implemented by the R package \textit{LatticeKrig} \citep{LKPackage}. We test the performance of LatticeKrig with smoothness parameters $\nu=0.5$ and $1$, and spatially autoregressive weights AW=$4.1$, $5.0$, and $6.0$, two levels of resolutions, and $15$ grids points at the coarsest resolution. For additional details of these parameters, we refer the readers to \citet{LK2015}. The autoFRK develops adaptive basis functions corresponding to different scales with thin-plate splines. Based on a maximum number $B$ of bases, it selects $B^*\le B$ bases by the Akaike's information criterion (AIC). The method can be implemented by the R package \textit{autoFRK} \citep{autoFRKPackage}. Here we set $B$ as $210$, $500$, $750$ and $1000$. 

We repeat the above process $100$ times and report the average performance of these methods with regard to their MSPEs, computational time (seconds) of covariance estimation and prediction, and the actual number of bases. These results are summarized in Table~\ref{tab:prec}. For FULL and SP$(210)$--$(1000)$, Table~\ref{tab:prec} only lists the results of predictions using the estimated covariance calculated by \textit{GpGp}. Based on the exact covariance estimation calculated by the \textit{ExaGeoStat}, FULL and SP$(210)$--$(1000)$ have MSPEs $0.22096$, $0.31711$, $0.27879$, $0.26108$, and $0.25124$, respectively. For LatticeKrig, we only listed the best result among all the parameter settings, i.e., the result for $\nu=1$ and AW$=6.0$. FRK has the lowest MSPE when \#Bases is small but requires longer time. Moreover, the decrease in its MSPE is limited when \#Bases increases. Furthermore, with multiresolution bases, autoFRK behaves well when \#Bases is small. However, its computational time rapidly increases and it cannot be applied when $B$ is larger than about $1000$. LatticeKrig uses more bases and consumes more time. The performance of our method is less superior when \#Bases is small because SPs cannot represent the full data locations well. However, the performance is drastically improved as we increase $k$. SP$(1000)$ provides a smaller MSPE with less computational time compared to FRK$(3)$ and LatticeKrig, both of which use more bases. Thus, the superiority of our method is verified. 
\begin{table}[!t]
\caption{Performance of various methods for the annual total precipitation data. FULL represents the prediction based on full data and estimated covariance obtained by \textit{GpGp}.  SP$(210)$--$(1000)$ represent low-rank approximations based on $210$--$1000$ SPs and estimated covariance obtained by \textit{GpGp}. FRK$(2)$--$(3)$ represent the FRK method with two and three levels of resolutions, respectively. Furthermore, autoFRK$(210)$--$(1000)$ represent the autoFRK method with $B=210$--$1000$. The approximation with the minimum MSPE is marked in bold.}
\vspace{-15pt}
\footnotesize
{   
\begin{flushleft}
	\begin{tabular}{l|cccccc}
			\toprule
			Methods & FULL &SP($210$) & SP($500$) & SP($750$) & \textbf{SP$(\textbf{1000})$} & FRK($2$) \\ 
			\midrule
			 MSPE & $0.22097$ & $0.31694$ & $0.27897$ & $0.26138$ & $\textbf{0.25159}$ & $0.28670$ \\ \hline
			 Time (s) & $159$ & $9$ & $10$ & $11$ & $\textbf{13}$ & $15$  \\ \hline
			 \#Bases & $7000$ & $210$ & $500$ & $750$ & $\textbf{1000}$ & $210$   \\ 
			\bottomrule	
			Methods & FRK$(3)$ & autoFRK($210$) & autoFRK($500$) & autoFRK($750$) & autoFRK($1000$)  & LatticeKrig \\ 
			\midrule
			MSPE & $0.27809$ & $0.30461$ & $0.26883$ & $0.25975$ & -- & $0.29958$  \\ \hline
			Time (s) & $383$ & $6$ & $13$ & $21$ & -- & $51$  \\ \hline
			\#Bases & $1936$ & $205$ & $499$ & $748$ & -- & $1258$  \\ 
			\bottomrule	
	\end{tabular}
 \vspace{-20pt}
	\end{flushleft}}
\label{tab:prec}
\end{table}

\subsection{Total column ozone data}
In the second application, we demonstrate the feasibility of application of various methods on a large and nonstationary dataset, which includes $173,405$ observations of the level-$2$ total column ozone for October 1, 1988, along with their locations. The data were collected and preprocessed by NASA and were discussed by \citet{FRK2008} and \citet{meng2020}. The region of interest $\mathcal{X}$ has a regular shape. We fit the data with $Z\sim \mathrm{GP}(\mu,c)$, where $\mu$ is a constant and $c(\mathbf{x}_i,\mathbf{x}_j)=\sigma^2\Gamma(\nu)^{-1}2^{1-\nu}(\|\mathbf{L}\mathbf{x}_i-\mathbf{L}\mathbf{x}_j\|_2)^{\nu}\mathcal{K}_{\nu}(\|\mathbf{L}\mathbf{x}_i-\mathbf{L}\mathbf{x}_j\|_2)$; here, $\mu$, $\sigma^2$, $\nu$, and the $2\times 2$ symmetric matrix $\mathbf{L}$ are the parameters to be evaluated.

We divide the data into two parts: training data with $n=150,000$, and testing data with $n_t=23,405$. We evaluate the covariance function using the R package \textit{GpGp}. The SPs of $k=189$, $500$, $1361$, and $1755$ are generated for facilitating better comparison with the other methods. We fitted the low-rank approximation based on the estimated covariance and SPs induced by the training data and performed predictions on the testing data. For the total column ozone data, only FRK and LatticeKrig are considered for comparison as autoFRK is not applicable because of the large sample size. For implementing FRK, we set two and three levels of resolutions and Mat\'{e}rn kernel functions with $\nu=1.5$. For LatticeKrig, we explore the parameter settings described in the previous subsection, and we set $\nu=0.5$, AW$=4.1$, two levels of resolutions, and $15$ grid points at the coarsest resolution. Neither FRK nor LatticeKrig allow smoothly setting the number of bases. 


\begin{table}[!b]
\vspace{-10pt}
\caption{Performance of various methods for the total column ozone data. LatticeKrig$(2)$ and $(3)$ mean LatticeKrig with two and three levels of resolutions, respectively. The one with the minimum MSPE is marked in bold.}
\vspace{2mm}
\footnotesize
\centering
{   
	\begin{tabular}{l|cccc}
			\toprule
			Methods & FRK$(2)$ & FRK$(3)$ & SP($189$) & SP($500$) \\
			\midrule
			 MSPE & $105.94$ & $50.18$ & $181.12$ & $70.00$ \\ \hline
			 Time (s) & $299$ & $6521$ & $437$ & $462$    \\ \hline
			 \#Bases & $189$ & $1755$ & $189$ & $500$   \\ 
			\bottomrule	
			 Methods & SP($1361$) & \textbf{SP$\textbf{(1755)}$} & LatticeKrig$(2)$ & LatticeKrig$(3)$\\
			 \midrule
			 MSPE & $38.80$ & $\textbf{35.87}$ & $84.85$ & $38.31$\\ \hline
			 Time (s) & $531$ & $\textbf{563}$ & $132$ & $212$\\ \hline
			 \#Bases & $1361$ & $\textbf{1755}$ & $1361$ & $3840$\\
			 \bottomrule
	\end{tabular}}
\label{tab:ozone}
\end{table}

The average results over $100$ runs are summarized in Table~\ref{tab:ozone}. Similar to the first example, the MSPE of FRK slightly decreases but its computational cost dramatically increases with an increase in \#Bases. LatticeKrig, while being the fastest, does not have efficient bases. For instance, with the same \#Bases, the MSPE of LatticeKrig(2) is more than twice that of SP(1361). With comparable MSPEs, the \#Bases of LatticeKrig(3) is more than twice that of SP(1755). The basis functions induced by the SPs are more efficient. The MSPE of SP methods decreases significantly as \#Bases grows, with only a slight increase in computational time. In our method, the main time-consuming parts are the estimation of covariance parameters and the generation of the SPs. For the first issue, one can choose more efficient algorithms. To address the second issue, one can generate SPs from a random subset of the training data without affecting the space-filling property.

\section{Discussion}
\label{sec:discussion}
In this study, we derived the asymptotic performance of a low-rank GP prediction $\tilde{f}_{\hat c}$, investigated the influence of rep-points $D_k^*$ and estimated covariance $\hat c$ on the convergence rate of $\tilde{f}_{\hat c}$. Using the concept of energy distance, we demonstrated that the distribution of rep-points should be as close as possible to that of the full data locations. This insight motivates the utilization of SPs, which have space-fillingness but also can best mimic the full data. We set the order of $k$ such that the convergence rate of $\tilde{f}_{\hat c}$ is equal to that of $\hat f_{\hat c}$ under certain regularity conditions. These conditions are exemplified in Section S3 to facilitate comprehension. With $\hat \gamma$ indicating the smoothness of $\hat c$, we provided the convergence rates of $\tilde{f}_{\hat c}$ and $\hat f_{\hat c}$ when sufficient rep-points are given, i.e., $V(\tilde{f}_{\hat c}-f)=V(\hat{f}_{\hat c}-f)=O_p[n^{-\{\min(\hat\gamma,\gamma)-1\}/\hat\gamma}]$. This indicates that the closer the $\hat \gamma$ is to $\gamma$, the faster is the convergence. Moreover, we demonstrated the performance improvement of $\tilde{f}_{\hat c}$ by adjusting the order of the nugget effect. We demonstrated the value of $\gamma$ and confirmed our theoretical results via extensive numerical studies. Using two examples of real data, we validated the prediction of $\tilde{f}_{\hat c}$ based on SPs by comparing our results with the results obtained using other existing low-rank approximation methods.

In some of the existing methods, such as FRK and LatticeKrig, multiresolution bases are desirable to capture information at different scales. In this way, they can better represent the local dependency and evaluate the covariance. In this study, we used the bases induced by SPs without multiresolution capability because we fixed the covariance and focused only on the prediction accuracy. We preferred to consider the tasks of covariance estimation and prediction separately because they require different configurations of rep-points. Literature in the field of spatial designs suggests future avenues, e.g., samples in clusters behave better for covariance estimation, whereas widely distributed samples are better for prediction \citep{zhu2005spatial,zhu2006spatial,barbian2017spatial}. Therefore, one possible extension of our work is the configuration of rep-points for covariance estimation.
\vskip -40pt
\begin{center}
{\bf Supplementary Materials}
\end{center}
\vskip -10pt
A supplement to the main manuscript, including proofs for Theorems~\ref{thm:tildefc}--\ref{thm:adjuestedtau}, an illustrative example for regularity conditions [R1]--[R3], and additional simulated results, is provided in the online supplementary materials. (.pdf file)


The R code, data, and instructions for reproducing the results in the article are available at the GitHub repository: \href{https://github.com/SpatialTemporalStats/LRGPSP_Reproducibility_Materials}{SpatialTemporalStats/LRGPSP\underline{ }Reproducibility\underline{ }Materials}.
\vskip -40pt
\begin{center}
{\bf Acknowledgement}
\end{center}
\vskip -10pt
The authors thank the review team for comments that improved the content of this paper. This research was supported by Beijing Natural Science Foundation (No. Z200001), China Scholarship Council, Renmin University of China and by the King Abdullah University of Science and Technology.

\vskip -40pt
\begin{center}
{\bf Disclosure Statement}
\end{center}
\vskip -10pt
The authors report there are no competing interests to declare.



\vspace{-0.5cm}

\baselineskip 25 pt

\bibliographystyle{asa}

\bibliography{ref1}
\end{document}